\RequirePackage{fix-cm}
\documentclass[smallextended]{svjour3}       
\smartqed  

\usepackage{subfigure}
\usepackage{graphicx}
\usepackage{amssymb}
\usepackage{algorithmic}
\usepackage{algorithm} 
\usepackage{theorem}
\usepackage{subfigure}

\theoremheaderfont{\upshape\bfseries}
\theorembodyfont{\itshape}
\theoremstyle{plain} 
\newtheorem{Pro}{Proposition}

\newtheorem{Lem}{Lemma}
\newtheorem{Cor}{Corollary}

\begin{document}

\title{Measurement-Adaptive Cellular Random Access Protocols
\thanks{This work was supported by Alcatel-Lucent Bell Labs, Stuttgart, Germany.}
}


\author{Anastasios Giovanidis         \and Qi Liao \and S\l awomir Sta\'nczak}


\institute{A. Giovanidis \at
              INRIA - TREC, 23 avenue d'Italie, CS 81321, 75214 Paris Cedex 13, France \\
             \email{anastasios.giovanidis@inria.fr}           
           \and
           Q. Liao \at
              Fraunhofer Institute for Telecommunications, Heinrich Hertz Institute (HHI), Einsteinufer 37, 10587 Berlin, Germany\\
               \email{qi.liao@hhi.fraunhofer.de}
              \and
              S. Sta\'nczak \at
              HHI and Heinrich-Hertz-Lehrstuhl f\"ur Informationstheorie und theoretische Informationstechnik, Technische Universit\"at Berlin, Einsteinufer 27, 10587 Berlin, Germany\\
              \email{slawomir.stanczak@hhi.fraunhofer.de}
}

\date{Received: date / Accepted: date}

\maketitle

\begin{abstract}
This work considers a single-cell random access channel (RACH) in cellular wireless networks. Communications over RACH take place when users try to connect to a base station during a handover or 
when establishing a new connection. We approach the problem of optimal coordination of user actions, 
taking into account a dynamic environment (channel fading, mobility, etc.). Within the framework of Self-Organizing Networks (SONs), the system should 
self-adapt to such environments without human intervention. To do so certain information should be gathered at the base station.
For the performance improvement of the RACH procedure, we propose protocols which exploit information from measurements and user reports in order to estimate current values of the system unknowns and broadcast global action-related values to all users. The protocols suggest an optimal pair of user actions (transmission power and back-off probability) found by minimizing the drift of a certain function. Numerical results illustrate the great performance benefits at a very low or even zero cost in power expenditure and delay, as well as the fast adaptability of the protocols to envoronment changes.

\end{abstract}

\section{Introduction}
\label{intro}

Random multiple access schemes have traditionally played an important role in wireless communication systems. 
Their use has been established especially in cases of bursty source traffic, where a multiplicity of users 
requires access from a central receiver. Starting with the ALOHA protocol \cite{Abramson}, several modifications have been suggested in the years to come aiming at performance improvement \cite{EphreTutor}. A very common application is in wireless LANs, such as the IEEE 802.11 protocol (see \cite{BianchiRACH}, \cite{GuptaRACHINFOCOM}, \cite{SharmaINFOCOM06} and references therein). The random access channel (RACH) is also included in the 3rd Generation Partnership Project (3GPP) as an important element within the Long Term Evolution (LTE) of cellular systems \cite{3gppRACH}, \cite{3gppson}, \cite{GPPTS36321}.

In the case of wireless cellular networks, a very limited frequency resource is reserved for the 
cases when a user requests for access from a base station (BS) or in order to be synchronized for uplink/downlink data transmission. 
RACH communications further occur during the hand-over phase \cite{Kiki11}, because 
of user mobility, or when a user is (re-)initiating some new service.  RACH channel can be used as well during the load balancing procedure \cite{GiovWSA12}, when cell-edge users are pushed to migrate to a neighboring BS after modification of the cell individual offset.

Due to limited resources, connection failure can occur in cases when the system is not well adapted to the incoming traffic. 
Consider for example large spaces in cities where occasionaly a vast amount of requests for service can be demanded, although 
normaly the system is not heavily loaded (e.g. metro stations, market streets, stadiums, city squares, areas close to concert and conference halls etc.).
In such places, it is very common that the system fails to support the service for all users and one of the reasons can be high collision rate in the RACH channel. 
It is thus necessary, within the context of Self Organizing Networks (SON) \cite{3gppson}, \cite{OsterboRACH} that the system can adapt to abrupt environmental changes that influence its functionality.  
Thus the RACH self-optmization problem is identified as an important case in the LTE standardization process \cite[paragraph 4.7]{3gppson}.

Unfortunately, in all such cases, the cellular system has almost zero user-specific information. Each BS can however broadcast certain information with 
cell-specific access details \cite{SelfOpt3G}, which allow the users to adapt their operation. Furthermore, carrier sensing as understood in the 802.11 is here not possible, which provides limitations to the design of high performance protocols. This is because, the possibility for a user to sense whether the channel is idle or not, is not provided and collision events cannot be avoided.

The procedure is called random access, due to the fact that the users access the channel in 
a random fashion. In the ALOHA case, when more than one user transmit simultaneously and 
their signals are detected we say that a collision occurs and all efforts are considered unsuccessful. LTE standardization, instead, provides the possibility for each user to randomly choose over a common pool of orthogonal frequencies \cite {3gppRACH} and a collision takes place 
when at least two users make the same choice during the same transmission interval.
After a failure, each source enters a back-off mode. The period of user silence is 
usually chosen having an exponential distribution but other possibilities can be used 
when such choice is adapted dynamically. This back-off time can generally be modeled in the slotted case by a per slot probability of transmission, less than 
$1$. Using this technique, an increase in throughput is achieved at the cost of additional delay.
Furthermore, since the detection or not of a user signal is also critical for the success, an important parameter is the transmission power of each user as well.

In short, the access (back-off) probability and the signal power are the two user actions.
An interesting idea to improve the decision making is to make certain global information of the system state available by broadcasting it from the base station. This is compatible with LTE standards where other type of information is already considered as globally known \cite{3gppRACH}. The information should 
represent the current system situation, so that users may adapt their actions dynamically. 
In this way the delay-throughput tradeoff can be enhanced. The cost is certain signaling and 
computations for the updates at the BS side. Furthermore, the BS
should have a way to gather relevant empirical information from its environment, related to the 
RACH functionality.

Based on the above idea, the current work suggests a dynamically adaptive RACH protocol for the cellular systems focused on LTE design. Empirical information is gathered through 
measurements and user reports. After certain processing at the BS side global system parameters are broadcast to users who require access. 
The protocol suggested, which is based on adaptation of the system to changes in the environment, further guarantees near-optimal performance related to a certain performance metric. 

\subsection{Related Literature}

Bianchi \cite{BianchiRACH} has been the first to provide a precise performance analysis for a random access protocol, which 
uses exponential back-off times. His approach considers a saturated system model, where the number of users is kept fixed to $N$ and 
all have a packet to send at each time slot. The results are based on the key approximation that the collision probability of a packet transmitted 
is constant and independent, which decouples the evolution of the system to $N$ 1-dimensional Markov Chains. 

A different approach has been suggested by Sharma et al. \cite{SharmaINFOCOM06}, where more general 
back-off strategies (generalized geometric) are considered for the IEEE 802.11 protocol 
in order to take service differentiation into account.
One of the major differences 
is that the system state is described by the current number of users per effort, while the collision probability is not independent per 
user.

First suggestions for dynamically controlling multiple access protocols can be 
found in Hajek and van Loon \cite{HajekRACH} as well as Lam and Kleinrock \cite{LamKlein75}. More recently Markov Decision Processes (MDPs) have been used in \cite{delAngel04} to derive optimal power and back-off policies for a set of backlogged users in slotted ALOHA random access systems. Cases of unknown user number have also been taken into account.

Gupta et al \cite{GuptaRACHINFOCOM} have recently suggested a dynamic back-off adaptation mechanism, where contention is regulated by broadcasting 
a so called contention level to the users. This is similar to the idea used in our approach.
Works of particular interest are also those of Liu et al \cite{Liu09} and Cheung et al \cite{Cheung10} which use the framework of utility-optimization for the optimal choice of transmission probabilities.

Channel-aware scheduling approaches in conjunction with random access mechanisms (which do not find application here due to the lack of such information in the system) include \cite{DimicTutor}, \cite{TongTutor}, and more recently \cite{BorstRandomAccNEW}. 

How random access works in the 3GPP-LTE systems is thoroughly described in \cite{SelfOpt3G}, 
where certain suggestions are presented, related to a self-organizing mechanism with information exchange between users and the Base Station. Investigations on the 
RACH power control include \cite{LeeRACH} and references therein, whereas an analytical framework for RACH modeling and optimization is given in \cite{YilmazRACH}.

Finally, rather interesting for the CSMA/CA case is the dynamic adaptation mechanism suggested in \cite{IdleSense} where users adapt their time window based on measurements and estimation of the average number of idle time slots of the random access channel. It involves an AIMD (Additive Increase Multiplicative Decrease) rule for the updates. Unfortunately, such a technique cannot be directly applied to the cellular system due to the unavailability of the sensing mechanism, it can however give ideas for application of a similar mechanism for the power updates.

\subsection{Contributions and Outline}

We investigate a saturated system model, where a number of $N$ users are always present within a wireless cell and try to gain access to the Base Station. An effort is successful when the user transmits a certain sequence, which is detected at the Base Station and at the same time no collision occurs. The event of collision will happen when the transmitted sequence of another  user is also detected. Furthermore, LTE standards allow for orthogonal sequences randomly chosen by the users, so that even when two user signals are detected, access to both may be granted.

In our analysis the miss-detection probability and collision probability are left as unknown variables. However, higher power increases the chances for detection and reduces collision probability, whereas use of access (otherwise back-off) probabilities reduces the collision events. Transmission power and access probability are the user action pair.

After description of the action space and state space, the transition probabilities are given and the evolution of the system is described by a Markov Chain. Furthermore, the event of dropping, when the users exhaust the maximum number of efforts allowed, plays an important role. Unfortunately, due to the unknown expression for the success probability no steady-state analysis is possible. The above are analytically presented in Section \ref{Section1}.

What we can do however, is to choose the actions myopically optimal, in the sense that they optimize the expected change in one time-slot for some function of the  state space. For this we introduce in our analysis the drift of a delay-related function. To motivate further our formulation, it is shown in the Appendix how the solution of the drift minimization problem is related to the solution of an ideal Markov Decision Problem for optimal performance in the steady-state. Our problem formulation is found in Section \ref{Section2}.

To solve the problem online a protocol is introduced. Its steps are presented in Section \ref{Section3}. The BS collects measurements as well as user reports to estimate the unknown probabilities (miss-detection, contention, success) at the Base Station side, as well as the current number of users, which is actually unknown in a real system. After solution of an optimization problem and a close-loop control problem, the BS broadcasts two values, the current \textit{contention level} and the current \textit{power transmission level}, so that the users can update 
their action pair. 

Numerical simulations for the performance of the protocol in a wireless cell are presented in Section \ref{Section4}. Advantages and trade-offs in delay and power expenditure are discussed and explicitly illustrated in plots. Finally, Section \ref{Section5} concludes our work.

\section{System Model}
\label{Section1}

\subsection{General Description}

We consider an arbitrary but fixed total number of $N$ users labeled by $n=1,\ldots,N$ trying to randomly obtain access to a cell Base Station (BS) 
over the wireless channel. The time is slotted, with each slot interval normalized to $1$ and indexed by $t$. At each time slot all users belonging to the user set have 
the possibility to access the channel by transmitting a preamble sequence (as specified in the LTE standards). There are two criteria that determine the success of an attempt. 

\begin{itemize}
\item \textit{The Signal-to-Noise Ratio (SNR) at the BS exceeds a predefined detection threshold $\gamma_d$}. If the SNR is below the threshold, we assume that a miss-detection occurs and the user has to retry. 
The \textbf{Detection Miss Probability} (DMP) can be written as the probability of an outage event
\begin{eqnarray}
 \label{DMPn}
Q^o_n\left(p_n,t\right) & = & \mathbb{P}\left[SNR_n\left(p_n\left(t\right),h_n\left(t\right)\right) \leq \gamma_d\right]
\end{eqnarray}
where $p_n$ is the chosen transmission power and the probability is taken over the random channel quantity denoted by $h_n$ and is i.i.d. over time $t$. In general we will consider that the BS does not approximate somehow the expression for outage. This is reasonable since the information over the user positions and the exact fading statistics is not known a priori.

\item \textit{No collision of transmitted signals occurs}. Typically in the slotted ALOHA protocol \cite{Abramson}, when more than one user attempts to access the channel during the 
same time slot a collision occurs and all affected users
have to repeat the effort.  In more recent wireless protocols, such 
as those suggested in LTE standards \cite{GPPTS36321}, a pool of orthogonal sequences (e.g. Zadoff-Chu) is made available to all users. Each user chooses 
one sequence from this set randomly (uniform distribution) and the probability of collision can be made less than $1$ when two users transmit simultaneously. 

In our model, the probability of collision is conditional on the transmission and the detection of signals at the BS side. That is, a user may collide only if he transmits at time slot $t$ and his signal is detected. Assuming that $N$ users transmit at time slot $t$ with transmission probability vector $\mathbf{1}_N := \left[1,\ldots,1\right]$ and $k$-out-of-$N$ (we write $k\setminus N$) are detected, the overall \textbf{Collision Probability} (CP) - the probability that at least one collision occurs - is an increasing function of both $N$ and $k$

\begin{eqnarray}
 \label{Col1}
Q^c\left(N,\mathbf{1}_N,k,t\right) 
\end{eqnarray}
As in the case of the DMP we consider that the base station does not have an exact closed form expression to calculate the CP and the above quantity is in general unknown.
\end{itemize}

\subsection{Action Space}
There are \textit{two actions} that user $n$ can take for transmission at time slot $t$.

\begin{itemize}
\item The choice of the \textbf{transmission power level} $p_n\left(t\right)$, which influences the detection of the 
transmitted signal at the BS, as shown in (\ref{DMPn}) and eventually the collision probability (through the number of detections $k$). In general $Q_n^o$ exhibits a monotone decreasing behavior with respect to power.

\item The choice of the \textbf{access (or transmission) probability} $b_n\left(t\right)$ per user, at a given slot $t$. This influences the number of simultaneously transmitting users in the cell and therefore directly affects the collision probability in (\ref{Col1}). The \textit{back-off} probability simply equals $1-b_n\left(t\right)$.
\end{itemize}

The set of actions for the entire system of $N$ users at $t$ is denoted by the $2N$-dimensional vector 
$\mathbf{A}\left(t\right):=\left(\mathbf{b}_N(t),\mathbf{p}_N(t)\right)$. The action space per time-slot is denoted by $\mathbb{A}$ and is the Cartesian product $\left[0,1\right]^N\times\left[0,P_1\right]\times\ldots\times\left[0,P_N\right]$, where $P_n$ is 
a given individual user power constraint per slot. Furthermore, $\tilde{\mathbf{A}}=\left\{\mathbf{A}(1),\ldots,\mathbf{A}(t),\ldots\right\}$.

Until the end of the subsection, we provide a discussion on the influence of choice for the back-off probability. In the definition (\ref{Col1}) no 
back-off action is taken, $b_n\left(t\right) = 1$, $\forall n$ and all users transmit simultaneously. On the other hand, assigning $b_n\left(t\right)\leq 1$ to some users, displaces the transmissions in time and the effect of collision is mitigated. Since less than $N$ users simultaneously compete for the access of the medium in some slot $t$, the collision probability is reduced.  
This can also be shown analytically. 

The overall collision probability of $N$ users present within the cell, with access probability $N$-length vector $\mathbf{b}_N$, $b_n\leq 1$ and exactly $k$ users detected, equals 

\begin{eqnarray}
 \label{ColisionBK}
Q^{c}\left(N,\mathbf{b}_N,k,t\right) & = & \sum_{J=0}^{N} Q^c\left(J,\mathbf{1}_J,k,t\right)\cdot Q^t\left(\mathbf{b}_N,J\setminus N\right)
\end{eqnarray}
where $Q^t\left(\mathbf{b}_N,J\setminus N\right)$ is the probability that - given a 
probability vector $\mathbf{b}_N$ - exactly $J$-out-of-$N$ users in the cell transmit. The equality follows from the total probability theorem, since the union of events $J=0,\ldots,N$ transmissions exhaust the sample space. The transmission probability of $J\setminus N$ users equals 

\begin{eqnarray}
Q^t\left(\mathbf{b}_N,J\setminus N\right) = \sum_{l=1}^{L\left(N,J\right)}\prod_{i=1}^J b_{q_l^{J.i}} \prod_{j=1}^{N-J} (1-b_{\hat{q}_l^{J.j}})\nonumber
\end{eqnarray}
where the summation over $l$ is taken over all possible 
$L\left(N,J\right) = \left(
\begin{tabular}{c}
$N$\\
$J$
\end{tabular}
\right)$ 
combinations (sampling without replacement) of $J$ users transmitting and $N-J$ users remaining silent, $q_l^{J.i}$ is the index of user $i$ belonging to combination $l$ that transmits and $\hat{q}_l^{J.j}$ is the index for the user $j$ that does not transmit.

 \begin{Pro}
  \label{Pro1}
Given $\mathbf{b}_N<\mathbf{1}_N$ (the inequality means that $b_n<1$ for at least one $n$) and exactly $1\leq k\leq N$ detections, we have that
 \begin{eqnarray}
  \label{Bigger}
 Q^c\left(N,\mathbf{b}_N,k,t\right) & < & Q^c\left(N,\mathbf{1}_N,k,t\right) 
 \end{eqnarray}
 \end{Pro}
 
 \begin{proof}: The events $J=0,\ldots,N$ exhaust the sample space and we have that their probability sum equals 
$\sum_{J=0}^N Q^t\left(\mathbf{b}_N,J\setminus N\right) = 1$. Furthermore, for $J<k$ it holds $Q^c\left(J,\mathbf{1}_J,k,t\right)=0$  since there cannot be more detections than transmissions.
The higher the number of transmissions, the higher the collision probability, which means $Q^c\left(J,\mathbf{1}_J,k,t\right)\leq Q^c\left(N,\mathbf{1}_N,k,t\right)$, $\forall J$ and the inequality is strict for $J<k$. From (\ref{ColisionBK}) we have

\begin{eqnarray}
\label{proof1}
Q^{c}\left(N,\mathbf{b}_N,k,t\right) & < & Q^c\left(N,\mathbf{1}_N,k,t\right)\cdot \sum_{J=0}^{N} Q^t\left(\mathbf{b}_N,J\setminus N\right)\nonumber\\
& = & Q^c\left(N,\mathbf{1}_N,k,t\right)\nonumber
\end{eqnarray}
which concludes the proof.
 \end{proof}

\subsection{Success Probability, Failure Event and Dropping}

From the above, success of a transmission is an event which occurs when (i) a user transmits, 
(ii) the user signal is detected and (iii) no collision occurs. In the use of orthogonal 
sequences/preambles, it suffices that no two users sharing the same sequence collide. In general, conditioned that a user transmits, the \textbf{Success Probability} (SP) equals

\begin{eqnarray}
 \label{SuccProb}
 Q^s_n\left(N,k,\mathbf{b}_N,p_n,t\right) & = & \left(1-Q^o_n\left(p_n,t\right)\right) \cdot \left(1-Q^c\left(N,\mathbf{b}_N,k,t\right)\right)
\end{eqnarray}
Observe, that the success probability of a single user does not depend only on his own action set $\left(b_n,p_n\right)$, but also on the choices of access probabilities of the other users, as well as the number of detected users $k$. The latter is further dependent on the transmission power chosen for $j\neq n$, so we can instead write 

\begin{eqnarray}
Q^s_n\left(N,\mathbf{b}_N,\mathbf{p}_N,t\right)
\end{eqnarray}

In the case of an unsuccessful effort the user may retry. Each user is constrained to at most $M$ \textit{access efforts} and the efforts are indexed by $m$. After $M$ unsuccessful efforts the user is considered discarded and replaced by a new-coming one, 
so that the total user number in the system always remains equal to $N$. The same holds when a user leaves the system after success. Therefore, we say that the system is \textit{saturated}. The number of users at effort $m$ in time slot $t$ is denoted by $X_m\left(t\right)$ and from the above it follows that

\begin{eqnarray}
 \label{sumuser}
\sum_{m=1}^M X_m\left(t\right) = N, & \forall t.
\end{eqnarray}
We occasionaly write in the following that a user at effort $m\in\left\{1,\ldots,M\right\}$ belongs to \textit{user class} $m$.

\subsection{System States and Transition Probabilities}

We define the state of user $n$ at slot $t$ as the current transmission effort $S_n\left(t\right) \in \left\{1,\ldots,M\right\}$, whereas the 
system state as the $N$-dimensional vector
\begin{eqnarray}
 \label{SystemState}
\mathbf{S}\left(t\right) = \left(S_1\left(t\right),\ldots,S_N\left(t\right)\right).
\end{eqnarray}
Altogether, there are $M$ different user states and $M^N$ different system states (e.g for a cell with $10$ users and maximum $5$ efforts, the number is approximately $10$ million). 
The entire state space is denoted by $\mathcal{S}$. 
It is easy to verify that the system state forms an $N$-dimensional Markov chain.

We group the transitions 
for each user into (a) returning to state $1$ in case of transmission and success, 
(b) moving to the next effort in case of transmission and failure and 
(c) backing-off and remaining in the same state. The expressions for the transition probabilities are given below. (Dependence of the functions on other parameters except the time index is omitted for brevity of presentation.)

\begin{itemize}
 \item For $1\leq m< M$:
\begin{eqnarray}
 \label{TrPm1}
\mathbb{P}\left[S_n\left(t+1\right) = 1|S_n\left(t\right)\right] 			& = & b_n\left(t\right)\cdot Q_n^s\left(t\right)\\
\label{TrPm2}
\mathbb{P}\left[S_n\left(t+1\right) = S_n\left(t\right)+1|S_n\left(t\right)\right] 	& = & b_n\left(t\right)\cdot  \left(1-Q_n^s\left(t\right)\right)\\
\label{TrPm3}
\mathbb{P}\left[S_n\left(t+1\right) = S_n\left(t\right)|S_n\left(t\right)\right] 	& = & 1-b_n\left(t\right)
\end{eqnarray}
\item For the user boundary state $m=M$:
\begin{eqnarray}
 \label{TrPMm1}
\mathbb{P}\left[S_n\left(t+1\right) = 1|S_n\left(t\right)=M\right] 			& = & b_n\left(t\right)\\
\label{TrPMm2}
\mathbb{P}\left[S_n\left(t+1\right) = M|S_n\left(t\right)=M\right] 			& = & 1-b_n\left(t\right)
\end{eqnarray}
A user in state $M$ will either back-off, in which case he remains in the same state, or transmit. When a user transmits, he will either succeed or fail. 
In both cases the next state is set to 1, the user is removed from the system and is replaced by a new one so that the total number is always equal to $N$.
The transition probabilities in (\ref{TrPMm1})-(\ref{TrPMm2}) for $m=M$ coincide with those for $m<M$, given by (\ref{TrPm1})-(\ref{TrPm3}) when 
$Q_n^s\left(t\right)=1$. In other words, to keep the system saturated, the Markov Chain evolves as if transmission at state $M$ always results in success.
\end{itemize}
This is why, it is further important for the analysis to specify the user \textbf{Dropping Probability} (DP)
\begin{eqnarray}
 \label{TrPdrop}
Q^d_n\left(N,\mathbf{b}_N,\mathbf{p}_N,M,t\right) & = & b_n\left(t\right) \cdot \left(1-Q^s_n\left(t\right)\right)\cdot \mathbb{P}\left[S_n\left(t\right)=M\right]
\end{eqnarray}
If the exact expressions for the DMP and CP were available, it would be possible to calculate the steady state probabilities of the system, by forming the $M^N\times M^N$ transition probability matrix and using the Perron-Frobenius theory  \cite[Ch. 2 and 8]{NonnegativeMatrix}. Since the number of states is finite, 
and for each user the probabilities (\ref{TrPm1})-(\ref{TrPm3}) and (\ref{TrPMm1})-(\ref{TrPMm2}) sum up to $\sum_{m=1}^M \mathbb{P}\left[S_n\left(t+1\right)=m|S_n\left(t\right)\right] = 1$ (stochastic matrix),  
a steady state with probability sum equal to $1$ always exists, although certain states may be transient and have zero probability.

\section{Problem Statement as Drift Minimization}
\label{Section2}

Since the exact expressions for the detection miss probability $Q_n^o$ as well as contention probability $Q^c$ are unknown (hence the success probability $Q_n^s$, which appears in (\ref{TrPm1}) and (\ref{TrPm2})), it is not possible to use the standard steady-state analysis as followed in \cite{TakKlein85}, \cite{Maglaris87}, \cite{Proutiere08}, \cite{Paschos07}, \cite{Klein75b} and \cite{Liu09} (among others) to derive long-term performance measures and optimize the system. Even if this would be possible however, the solution of a system of such an immense number of variables would bring difficulties (remember the number of 10 million variables for $N=10$ and $M=5$). The same problems are met in a Markov Decision Problem (MDP) formulation, as followed e.g. in \cite{LamKlein75} and \cite{delAngel04}.

Furthermore, in a realistic setting, we would like to propose a protocol, which takes into consideration the fact that within the wireless cell, users appear and leave the system after a while, whereas the fading situation changes unpredictably. These two factors greatly influence the miss-detection and collision probabilities, which do not remain fixed until infinity, but exhibit large fluctuations over time. This falls within the concept of SON's which should self-adapte and self-optimize the wireless system parameters as a reaction to such unpredictable changes from outside without human intervention. 

For the above reasons we make use of the notion of \textit{drift for the Markov Chain} under study, in order to achieve an improvement in the system performance by appropriate choice of actions. The idea of drift is commonly used in the literature of stability of systems with infinite states \cite{TassQue93}, \cite{TassMultihop}, \cite{NeelySat}, \cite{NeelyRouting}. In such cases, if we can find, for a given positive Lyapunov function, an action policy which keeps the drift negative for the entire state space - except possibly for some finite subspace - the system is guaranteed to remain stable. This comes from direct application of Foster's theorem (see \cite[Prop. 5.3(ii)]{Asmussen}). Intuitively the negative drift gives the function of states a tendency to decrease in expectation at each step, as long as it is outside the aforementioned subspace, so that in the long run the value a state can take will not be unbounded (and the stability is guaranteed).
In our case the state space is finite due to the finiteness of $M$. However, since the amount of users that exceed $M$ efforts are eventually dropped, stability of the system refers to keeping the number of dropped users finite. (Alternative application of the drift minimization to a problem with $M\rightarrow\infty$ and no dropping does not change much the policy and results).

The drift equals per definition, the expected change in the Lyapunov function from $t$ to $t+1$. By choosing an appropriate non-negative function of the system state $V\left(\mathbf{S}\left(t\right)\right)$ related to some performance criterion, we can choose actions that optimize performance at each time-slot. Since it is impossible to know how the system will evolve in future slots, and since expressions for DMP and CP are not available, the best thing we can do is to provide a one-step look-ahead (\textbf{myopic}) policy for the system, given its current state and measurements performed on time $t$, which estimate unknown parameters. Specifically, given that the system state at $t$ is $\mathbf{S}\left(t\right)$, the drift is defined as

\begin{eqnarray}
 \label{Drift}
D\left(V\left(\mathbf{S}\left(t\right)\right),\mathbf{A}\left(t\right)\right) & := & \mathbb{E}\left[V\left(\mathbf{S}\left(t+1\right)\right)-V\left(\mathbf{S}\left(t\right)\right)|\mathbf{S}\left(t\right)\right]
\end{eqnarray}
and is also a function of the action set $\mathbf{A}\left(t\right)$, since the actions control the system state transition probabilities $p_{s_t\rightarrow s_{t+1}}$.

The function $V$ to be used is the sum of user states and is linear. It can be rewritten as the sum of cardinalities of users at a state, weighted by their effort index. 
\begin{eqnarray}
 \label{FunctionVhere}
V\left(\mathbf{S}\left(t\right)\right) = \sum_{n=1}^{N} S_n\left(t\right) = \sum_{m=1}^{M} m\cdot X_m\left(t\right)
\end{eqnarray}

A user who is currently at a higher effort, contributes more to the function, than 
users at lower ones. By \textbf{minimizing} the drift of such function we wish to choose 
appropriate actions in order to have success with as few efforts as possible. This has following objectives: 
\begin{itemize}
\item keep a good trade-off between power consumption and delay until success per user 
\item diminish the proportion of users who are dropped 
\item maximize a notion of total system throughput
\end{itemize}
To understand the last point, observe that each user $n$ contributes a ratio $\frac{1}{m^*_n}$ to 
the total system throughput if $m^*_n\leq M$ efforts are required for success and contributes nothing if the user is dropped. Consider now as a single virtual user, the set $N$ of users in the network. By use of the Renewal-Reward theorem \cite{GiovSCC08}, the long-term throughput of such a virtual user (considering only number of efforts and not the total number of time-slots required including user silence slots) will be the ratio $\frac{N}{\mathbb{E}\left[V\left(S\right)\right]}$. Alternative Lyapunov function could change the objective of the minimization, giving emphasis to total delay or power consumption and can be understood as alternative formulations of the same general problem and solution methodology.

Let us consider state-dependent, rather than user-dependent actions, in the sense that all users who are at class $m$ in slot $t$ should make the same choice for transmission power and back-off. The specific drift expression can now be derived to yield

\begin{eqnarray}
 \label{DriftHere}
D\left(V\left(\mathbf{S}\left(t\right)\right),\mathbf{A}\left(t\right)\right) & = & \sum_{n=1}^N\left\{1\cdot\mathbb{P}\left[S_n\left(t+1\right) = 1|S_n\left(t\right)\right]+ \right.\nonumber\\
								&   & \left(S_n\left(t\right)+1\right)\cdot\mathbb{P}\left[S_n\left(t+1\right) = S_n\left(t\right)+1|S_n\left(t\right)\right]+\nonumber\\
								&   & \left.S_n\left(t\right)\cdot\mathbb{P}\left[S_n\left(t+1\right) = S_n\left(t\right)|S_n\left(t\right)\right]- S_n\left(t\right)\right\}\nonumber\\
								& \stackrel{(\ref{TrPm1})-(\ref{TrPMm2})}{=} & \sum_{n=1}^N {b_n\left(t\right)\cdot\left[1-S_n\left(t\right)\cdot Q^s_n\left(N,\mathbf{b}_N,\mathbf{p}_N,t\right)\right]}\nonumber\\
								& \stackrel{state \ dep.}{=} &  \sum_{m=1}^M X_m\left(t\right)b_m\left(t\right)\cdot\left[1-m Q^s_m\left(N,\mathbf{b}_N,\mathbf{p}_N,t\right) \right]
\end{eqnarray}

The drift minimization problem at each time slot $t$ is
\begin{eqnarray}
 \label{ProblemG}
\begin{tabular}{l l}
\textbf{min} 					& $D\left(V\left(\mathbf{S}\left(t\right)\right),\mathbf{A}\left(t\right)\right)$\\
\textbf{s.t.}					& $\mathbf{A}\left(t\right)\in\mathbb{A}$
\end{tabular}
\end{eqnarray}
A further motivation to pose the problem as a drift minimization is provided in the Appendix. It is shown that (\ref{ProblemG}) is a myopic solution of an MDP with objective the minimization of the 
expected Lyaponov function at the steady-state (for $t\rightarrow \infty$). For the formulation and solution of the MDP, the expression for $Q^s_n$, $\forall n$ should be available and the channel/user statistics should remain unchanged over the entire time horizon.
 
What is needed to solve the above problem per slot? It follows from (\ref{DriftHere}) that the 
following information should be available at the BS side:

\begin{enumerate}
 \item The cardinality $X_m\left(t\right)$ of users at each effort $m$.
 \item The current value of $Q^o_m\left(t\right)$ at each $m$.
 \item The current value of $Q^c\left(t\right)$.
\end{enumerate}
Using 2. and 3. and the product in (\ref{SuccProb}) the actual value of $Q^s_m\left(t\right)$ can be obtained. Although the BS does not know these values it may estimate the variables and with it approximate the objective function,
using \textit{measurements} related to channel and service quality, as well as 
\textit{information} reported directly by the user set. The goal is to use these estimates for optimization, in order to achieve significant performance gains, while keeping an additional overhead of exchanged information as small as 
possible. 

In this way, a sequence of problems with different 
numbers of users, contention and miss-detection probabilities can be solved over time, which help the cell to follow and adapt to dynamic unpredictable changes.
The steps of the proposed adaptive protocol are summarized in Table \ref{RACHalgo}.

\section{Five Steps of the Protocol}
\label{Section3}

Before proceeding to the algorithm, we first discuss over the action pair of access probabilities and transmission powers. Considering the access probabilities, we adopt the approach in \cite{GuptaRACHINFOCOM} (similar functions are also found in \cite{Liu09} and references therein), with per effort probability given by
\begin{eqnarray}
 \label{backoffFUN}
b_m\left(t\right) = \min\left\{\frac{f(m)}{L\left(t\right)}, 1\right\},\ \forall m.
\end{eqnarray}
Here and hereafter, $L$ is called \textit{contention level} and $f(m)$ is some fixed function of the transmission effort. In this way, a simple variable $L$ 
can simultaneously define the entire set of transmission probabilities. By choosing $f$ to be monotone increasing in $m$, priority is given 
to users with higher efforts, while such users obtain lower priorities when $f$ is strictly monotone decreasing. Typical back-off protocols follow the exponential rule, which reduces by half the probability 
of accessing the channel after each 
failure, so in this case $f(m) = 2^{-m+1}$ and $b_1 = 1/L$. Other possible choice could be $f(m) = m^{-a}$, $a\in\mathbb{R}_+$ (in this work and the simulations to follow the case $a=1$ is mostly used). Exponents $a>1$ will lead to an overly conservative system with large delays for users in higher states, whereas $a<<1$ tends to treat users of all classes with the same priority. 
In the following, the expression in (\ref{backoffFUN}) will sometimes be replaced by $b_m(t) = f(m)/L\left(t\right)$ and the constraint $b_m\left(t\right)\leq 1$ 
is taken into account in the constraint set of the minimization problem.

We consider, furthermore, the transmission power to vary per effort as a ramping function. This approach is often considered in practice 
(for related approaches, the reader is referred to \cite{SelfOpt3G} and references therein). The power level for 
the first effort is given by $p$ and for all efforts by the expression
\begin{eqnarray}
 \label{power}
p_m\left(t\right) & = & p\left(t\right) + \left(m-1\right)\cdot \Delta p,\ \forall m
\end{eqnarray}
where $\Delta p$ is the ramping step with a fixed (tunable) value. Thus, analogously to the case of the backoff probabilities, the 
vector of power actions can be defined by appropriate choice of the \textit{power level} $p\left(t\right)$ per time slot.

\subsection{Step 1: Measurements and User Reports}
When users attempt to randomly access the channel, we assume that the BS counts the overall number of detected user efforts, as 
well as the overall number of successful efforts. Given an observation window of length $W$, both the quantities 
depend on the time interval $\left[t-W+1,t\right]$ and are denoted by 
$N_d\left(t\right)$ and $N_s\left(t\right)$ respectively. 
Furthermore, after every successful effort, the users are assumed to \textit{report} to the BS, the total number of trials required to get access. 
In this way, the BS can keep track of the number of successes at effort $m$, within the observation window, denoted by
$n_{s,m}\left(t\right),\ \forall m$. 
The reports over the success state also provide information over the overall number of transmissions of users being at some state $m$. 
As an example, if within the observation period two users report success at effort $3$ and $2$ respectively, the BS can estimate the number of 
transmissions at state $m=1$ by $2$, at $m=2$ by $2$ and at state $m=3$ by $1$, without considering users that have yet not declared success, or are dropped. 
We denote these estimates by
$n_{t,m}\left(t\right),\ \forall m$ and their sum, which equals approximately the number of access efforts within the 
observation window, by $N_t\left(t\right) = \sum_{m=1}^M n_{t,m}$. 
Altogether, the set of gathered empirical information, updated per time slot, is represented by

\begin{eqnarray}
 \label{empirics}
\mathcal{I}\left(t\right):=\left\{N_d(t),\  N_s(t),\ N_t(t),\  n_{s,m}(t),\forall m,\ n_{t,m}(t),\forall m\right\}
\end{eqnarray}

\subsection{Step 2: Estimation of Unknowns in the Objective function}
Using the above counters, we can now approximate the unknowns in the expression (\ref{DriftHere}) that are briefly discussed in points 1. - 3. in the previous Section.

As far as the unknowns in 2. and 3. are concerned, the actual overall contention probability $Q^c\left(t\right)$ and per effort success probability $Q_m^s\left(t\right)$ in (\ref{SuccProb}), 
can be estimated by contention and success \textbf{rates}, an idea which has already appeared in \cite{SelfOpt3G}. Observe that 
the additional information about the per effort miss-detection probability $Q_m^o\left(t\right)$ cannot be deduced from the above measurements. What can be calculated, instead, is an 
overall rate of miss-detection (DMR), without differentiating between efforts, which we denote by $R^o\left(t\right)$.

\begin{eqnarray}
 \label{DMRa}
R^c\left(t\right) = 1-\frac{N_s\left(t\right)}{N_d\left(t\right)} & & (contention\ rate) \\
\label{DMRb}
R^s_m\left(t\right) = \frac{n_{s,m}\left(t\right)}{n_{t,m}\left(t\right)}, &  \forall m& (success\ rate\ per\ effort)\\
\label{DMRc}
R^o\left(t\right) = 1-\frac{N_d\left(t\right)}{N_t\left(t\right)} & & (miss-detection\ rate).
\end{eqnarray}

Regarding the number of users currently within the cell (discussed in 1.) and their estimation, 
we proceed as follows. Instead of attempting to find integer values, we consider arrival rates. 
As the total arrival rate of users 
we consider the ratio $\frac{N_s\left(t\right)}{W}$, which is the time dependent ratio 
of accepted users, divided by the observation window. The above is used under the assumption that only a very small 
fraction of the users are dropped throughout the process, so that almost all users appearing within the cell, will eventually 
have at some point a success. Taking dropped users into account requires an additive correcting term that may be deduced from empirical observations.

The window is considered long enough, so that the resulting success rates per state, $R^s_m\left(t\right)$ in (\ref{DMRb}),
approach the actual success probability per effort. These can replace the entries in the one-step transition 
probability matrix in equations (\ref{TrPm1})-(\ref{TrPm3}) and (\ref{TrPMm1})-(\ref{TrPMm2}). 
%
The steady state probability distribution is found by solving the system $\mathbf{\pi=\pi\cdot \hat{P}_M}$, where 
$\mathbf{\pi}$ is the row vector of the unknown probabilities for the $M$ states with $||\mathbf{\pi}||_1 = 1$ and  $\hat{P}_M$ is the transition probability matrix. 
The solution equals

\begin{eqnarray}
 \label{solva}
\pi_1\left(t\right) & = & \left({1+\sum_{i=2}^{M}\frac{b_1}{b_i}(1-R^s_1\left(t\right))\cdot\ldots\cdot (1-R^s_{i-1}\left(t\right))}\right)^{-1}\\
\label{solvb} 
\pi_m\left(t\right) & = & \pi_1\left(t\right)\cdot \left(\frac{b_1}{b_m}(1-R^s_1\left(t\right))\cdot\ldots\cdot (1-R^s_{m-1}\left(t\right))\right),\ 2\leq m\leq M.
\end{eqnarray}
The ratios of the unknown backoff probabilities $b_1/b_m$ are involved in the expression above. From the previous discussion $b_1/b_m = f(1)/f(m)$, 
which is known since the function $f$ is chosen a priori. With these observations and definitions at hand, we can estimate the user arrivals per effort according to

\begin{eqnarray}
 \label{EstimX}
\frac{X_m\left(t\right)}{W} \approx \pi_m\left(t\right)\cdot \frac{N_s\left(t\right)}{W}
\end{eqnarray}
where the $\pi_m$'s are the probabilities given by (\ref{solva}) and (\ref{solvb}).

\subsection{Step 3: Solving the Problem}
Once step 2 is performed, we can formulate the objective function to approximately solve problem (\ref{ProblemG}) and with it find the optimal actions per time slot. 
To this end, we break down the problem into two subproblems and propose two sub-algorithms based on the measurements and estimated quantities described above.

\textbf{Backoff Probability Problem}:
The objective function at the base station is estimated by

\begin{eqnarray}
 \label{ObjApprox}
\hat{D}\left(V\left(\mathbf{S}\left(t\right)\right),L\left(t\right)\right) := \frac{1}{L\left(t\right)}\cdot \left[\sum_{m=1}^M \pi_m\frac{N_s\left(t\right)}{W}f\left(m\right)\cdot \left(1-m\cdot R^s_m\left(t\right)\right)\right],
\end{eqnarray}
where the success probability $Q_m^s$ is substituted by the success rate $R_m^s$ in (\ref{DMRb}) and the average user number $\frac{X_m}{W}$ by the expression in (\ref{EstimX}). 
As long as such estimates are close to the actual values and are considered reliable, the BS can solve a problem with parameters adapted to the changing environment. 

When the expression in brackets above $\left[\ldots\right]$ is positive, the objective function is convex and decreasing in the contention level variable $L$ (behaves as $+\frac{1}{L}$). When $\left[\ldots\right]$ is negative, the objective is concave and increasing in $L$ (behaves as $-\frac{1}{L}$). Due to the monotonicity and concavity/convexity, the optimization will have as a result 
either maximum or minimum value of $L$ depending on the sign of the term inside the square brackets. 

In the following we provide the boundary values $L_{\min}$ and $L_{\max}$ of the domain of $L$. The lower bound on $L$ follows from the fact that 
all access probabilities are less than or equal to $1$:
\begin{eqnarray}
 \label{LowerL}
\frac{f\left(m\right)}{L\left(t\right)}\leq 1,\ \forall m &\Rightarrow & L\left(t\right)\geq L_{\min}:=\max\left\{f(m)\right\}.
\end{eqnarray}
To obtain an upper bound, we further provide a constraint on the probability of a time slot being idle (no user transmits). This probability is less than 
or equal to $\mathcal{A}$, which is a design factor for the system.

\begin{eqnarray}
 \label{Idleslot}
\mathbb{P}\left[IDN\right] = \prod_{m=1}^M\left(1-\frac{f(m)}{L\left(t\right)}\right)^{ \frac{X_m\left(t\right)}{W}} & \leq \mathcal{A} & \Rightarrow \nonumber\\ 
\sum_{m=1}^M \pi_m\frac{N_s\left(t\right)}{W}\cdot\log\left(1-\frac{f(m)}{L\left(t\right)}\right) & \leq \log(\mathcal{A}) &.
\end{eqnarray}
The left handside is increasing with $L$, thus the inequality provides an upper bound on $L$. If we solve (\ref{Idleslot}) for equality, we then derive the value of $L_{\max}$. Notice furthermore that, all values of $L$ within the interval $\left[L_{\min},L_{\max}\right]$ are feasible solutions of the contention level.

\begin{Pro}
\label{ProSum}
Considering the problem of minimizing $\hat{D}$ in (\ref{ObjApprox}) subject to the upper and lower bound constraints on $L$, the following 
necessary and sufficient optimality conditions hold:
\begin{itemize}
  \item if $\left[\sum_{m=1}^M \pi_m\frac{N_s\left(t\right)}{W}f\left(m\right)\cdot \left(1-m\cdot R^s_m\left(t\right)\right)\right]\geq 0$ then 
the optimal contention level equals $L_{\max}$ and is found by solving
\begin{eqnarray}
 \sum_{m=1}^M \pi_m\frac{N_s\left(t\right)}{W}\cdot\log\left(1-\frac{f(m)}{L^*\left(t\right)}\right)& = & \log(\mathcal{A})
\end{eqnarray}

  \item if $\left[\sum_{m=1}^M \pi_m\frac{N_s\left(t\right)}{W}f\left(m\right)\cdot \left(1-m\cdot R^s_m\left(t\right)\right)\right]< 0$ then the optimal 
contention level equals $L_{\min}$
\begin{eqnarray}
 L^*\left(t\right)= \max\left\{f(m)\right\}.
\end{eqnarray}
\end{itemize}
\end{Pro}

\textbf{Power Control Problem}: In order to identify optimal transmission levels, one could proceed along similar lines 
as above, to formulate an optimization problem, given the back-off probabilities $f(m)/L^*(t)$ and the contention rates 
$R^c(t)$ from (\ref{DMRa}). In order to determine the objective function based on (\ref{DriftHere}), 
which is denoted by $\tilde{D}\left(V\left(\mathbf{S}\left(t\right)\right),p\left(t\right)\right)$, the closed form expression 
for the detection-miss probability $Q_m^o\left(t\right)$ as a function of power may be necessary. It 
is however unlikely that the channel's fading behavior in practical systems can be accurately represented by a closed-form expression, especially since in the random access cellular system the user position is not known to the BS.

A different approach - which is adopted here - is to use a \textit{Multiplicative-Increase-Additive-Decrease} (MIAD) control rule, as in the case of congestion control protocols in 
TCP \cite{CongestionRef}. In this way, the BS reacts to the change of the estimated DMR stepwise, by increasing or decreasing the power level $p(t)$ per 
time slot, depending on the current value $R^o\left(t\right)$. We set two levels of action, a high detection-miss level $DMR^H$ and a low one $DMR^L$. The control 
loop then works as follows: When $DMR^H$ is exceeded, the power level is increased by multiplication with a tunable factor $1+\delta_1$. 
This action increases considerably the transmission power since miss-detection is highly non-desirable. When the ratio falls under the low level $DMR^L$, 
which is considered satisfactory for the system performance, 
the power is reduced in a conservative way, to reduce the energy consumption on the mobile devices, 
by subtracting a constant tunable amount of $\delta_2$. For instance $\delta_2$ can be set equal to the ramping step $\Delta p$ in (\ref{power}). 
The control loop is then 
described by the power updates

\begin{eqnarray}
 \label{PowerControl}
p^*\left(t\right) & = & \left\{
\begin{tabular}{l l}
 $p^*\left(t-1\right)\cdot\left(1+\delta_1\right)$, & if $R^o\left(t\right)>DMR^H$\\
 $p^*\left(t-1\right) - \delta_2$, & if $R^o\left(t\right)<DMR^L$
\end{tabular}
\right..
\end{eqnarray}

Obviously, updates on the per-effort ramping steps or user-specific power control could be much more beneficial instead of the update in the global power level $p\left(t\right)$. Furthermore, it is obvious that by varying $p\left(t\right)$ globally, power consumption will increase not only for users in higher efforts but also for those in their first effort, which may not be necessary. However, there are certain difficulties in providing a different type of feedback. Most 
importantly, there is no user channel state information available at the BS and channel adaptation is impossible. Furthermore, based on the possible approximations that - given the measurements and the reports - are suggested, only a global miss-detection rate $R^o$ can be estimated in (\ref{DMRc}) and no state-specific or user-specific rates (say $R_m^o$). We cannot approximate, in other words, the rate of miss-detection for a user at different states and as a result we cannot suggest different state-dependent power levels. Finally, state-dependent power control would increase considerably the feedback information broadcast to all users. For all the above reasons, the suggestion of the MIAD rule was considered more appropriate.

\subsection{Step 4 and 5: Broadcast of Information to the Users and Action Calculation}

The last two steps of the proposed algorithm involve the broadcasting of the action-related information to the users and the choice of 
appropriate actions by them. The broadcast information includes the pair consisting of the contention level and the power level

\begin{eqnarray}
 \label{BroadcastInfo}
\mathcal{J}\left(t\right) & := & \left\{L^*\left(t\right), p^*\left(t\right)\right\}.
\end{eqnarray}
Let us assume that the expressions in (\ref{backoffFUN}) and (\ref{power}) for the success probability and the power level per effort are known a priori to the mobile stations.
Since each user is aware of its current individual state $S_n\left(t\right)$, calculation of its own action pair is possible, according to
\begin{eqnarray}
 \label{actionsetuser}
A_n\left(S_n\left(t\right),\mathcal{J}(t)\right) = \left(b_n(t), p_n(t)\right) = \left(\frac{f(S_n(t))}{L^*(t)}, \ p^*(t) + S_n(t)\Delta p \right).
\end{eqnarray}

Note that if the required power and access functions ($f\left(\bullet\right)$ and the ramping step $\Delta p$) is not available at the mobiles, the BS could broadcast the entire vector of computed transmission powers and 
access probabilities to the users so that they choose the actions according to their current effort.

A remark considering implementation issues of such protocols is that the updates of these two levels are not expected to take place very frequently, but rather only at the rate of estimated change of user traffic and fading conditions. Furthermore, user reports and broadcast feedback from the BS is already suggested in standardization reports, so that the proposed protocol complies fully with the existing standardization literature \cite{3gppRACH}, \cite{3gppson}, \cite{GPPTS36321}, without introducing additional protocol information.


\section{Numerical results}
\label{Section4}

\subsection{Description of the simulations setting}

The proposed algorithm has been implemented in a single cell scenario. 
The 
users are randomly positioned, with a 
2D uniform distribution and the algorithm is evaluated for the cases of $N = 1,2,\ldots,14$ [users/time slot] present in the cell. 
Considering the transmission scenario, each user randomly chooses at each attempt one sequence,
 out of a pool of $10$ orthogonal sequences, 
and transmits with a chosen backoff probability and transmission power. The number $10$ is used for simulation purposes, whereas the actual number suggested in the LTE literature equals $64$; however not all users have access to the entire pool of sequences (see \cite{3gppRACH}) since the sequence allocation procedure is more complicated than the simple uniform choice we use here.

The signal experiences path loss due to the user-BS distance. Fast fading is initially not modeled (this will be considered in the second part of the Section for the power consumption evaluation) but the channel is considered AWGN with noise mean equal to $-133.2$ dBm. We have to note that in case fast-fading were also implemented, a further randomness in the channel would affect the signal detection and the protocol performance. To keep things simple, we consider first only the randomness of user positioning which affects the slow-fading coefficients - also unknowns during the procedure. The evaluation of the protocol's performance will not change much by adding more randomness factors.

An effort is successful when among the detected sequences there exists no pair that collides, in the sense that no two detected users choose the same sequence for transmission. 
A user is dropped when the effort fails at the maximum access effort $M=5$. 
After a success or an event of dropping, users are removed from the waiting-for-transmission list, 
and the same number of newly arriving users are added, each given a random position on the plane.  

Power and access probability for the users are computed per slot equal to the action pair in (\ref{actionsetuser}), for $f\left(m\right) = m^{-1}$.  The choice of exponent $-1$ is not conservative (whereas a higher exponent would be) while at the same time it takes class differentiation into account. Important is 
to notice that the expression of the function $f$ greatly affects the delay. On the other hand, the 
delay can be controlled by the parameter $\mathcal{A}$ which is system-operator-dependent and tunes the expected idle period. The set of values for the parameters of the system simulation are summarized in Table \ref{parameterA}.

Several factors for the 
protocol design have been left open for choice. One of them, as mentioned already, has been the desired idle probability $\mathcal{A}$. The higher factor  $\mathcal{A}$ is, the more the delay suffered by the system but the higher the benefits in dropping rate and power consumption are. Other important parameters are the steps $\delta_1,\ \delta_2$ and bounds $DMR^H$, $DMR^L$ of the MIAD rule, the access function $f$ and the adaptive window length $W$, which defines how fast should the protocol adjust to environmental changes. A summary of these tunable factors and how they are chosen within the simulation setting under consideration is provided in Table \ref{parameterB}.

\subsection{Comparison to a Fixed "open loop" Power Fixed Backoff protocol }

The suggested algorithm is compared to a scenario, where access probabilities and target power are held fixed, 
while the ramping step for the transmission power is predifined and same for all efforts. The fixed scenario is in other words an "open-loop" control scheme, with predefined constant  $\left(p,\Delta p\right)$. The choice for the fixed backoff probability in the comparison scenario, equals $\left[b_1,b_2,b_3,b_4,b_5\right] = \left[0.5,0.4,0.3,0.2,0.1\right]$ and is such that the average occurance of an idle slot is less than $\mathcal{A} = 0.05$, hence the channel is kept busy with user efforts for access during most of the time . In this sense, the comparison between the adaptive-protocol suggested and a fixed protocol is more fair for a tunable factor of  $\mathcal{A} = 0.05$ or less. How the average idle probability changes between  $\mathcal{A} = \{0.05, 0.25, 0.5\}$ and the fixed case can be seen in Fig. \ref{idle}. We refer the reader to the Parameter Table \ref{parameterA} for the actual values used throughout these simulations. The above fixed scenario is denoted by (FPFB) for Fixed Power Fixed Backoff. Two types of protocols are used for performance comparison:

\begin{itemize}
\item \textbf{Fixed Power Dynamic Backoff (FPFB) protocols}. In this case the "open loop" power control of the protocol is the same as in the fixed scenario FPFB case. The backoff mechanism adapts to measurements as suggested in the protocol description of this work (Paragraph 4.3, Backoff Probability Problem).
\item \textbf{Dynamic Power Dynamic Backoff (DPDB) protocols}. In this case both backoff and power are adapted as the protocol suggests in Paragraph 4.3. The backoff comes from the solution of the drift minimization problem, while the target power $p$ is adapted according to the MIAD rule.
\end{itemize}

\subsection{Performance Evaluation: Lyapunov Function and Number of Efforts}

The performance of the scheme and its comparison to the fixed scenario FPFB is initially 
illustrated in the plots of the performance metric in Fig.\ref{PM} and the plots of 
the average number of access efforts until success in Fig.\ref{effort}. The two figures show a close relation to each other, due to the choice of the specific Lyapunov function $V$. Since $V$ was chosen as the sum of user efforts, lower values translate into better performance for the protocol. In all six curves, our protocol outperforms the FPFB scenario in the 
metric chosen as well as in the average number of user efforts. Furthermore, all DPDB cases schow improved performance compared to FPDB, given a certain value of the parameter $\mathcal{A}$. The higher the value of tunable factor $\mathcal{A}$, the better the performance and the less the averge efforts required up to packet reception.

\subsection{Performance Evaluation: Delay, Power Consumption and Dropping Rate}

The three most important performance measures in random access that can illustrate the improvements of the 
suggested protocol are the total delay suffered by a packet until success (including backoff slots), the total 
transmission power used until success as well as the percentage of users dropped because the maximum number $M$ of efforts is exceeded.
These are shown in Fig.\ref{delayFP}, \ref{delayDP}, \ref{PFP}, \ref{PDP} and \ref{DRFP}, \ref{DRDP} respectively, for (a) the FPDB case and (b) the DPDB case.

From the plots, it is illustrated how an increase of the parameter $\mathcal{A}$ influences positively power consumption and dropping rate at the cost of delay. Furthermore, the DPDB schemes perform better than the FPDB schemes in terms of delay and dropping, but have a cost in power consumption. Altogether, the performance of the protocol is tunable, to the requirements of 
the service provider. If the delay is not an issue, power can be considerably saved and the number of users dropped is reduced. As long as delay becomes an issue, transmission power can still be saved by using only the FPDB protocols. The dropping rate is also improved in such a case. Dynamic backoff generally allows the system to remain stable - in the sense that the rate of dropped users does not tend to "explode" - for a higher value of $N$. The behavior of this measure also improves for higher $\mathcal{A}$, which is reasonable since allowing a higher idle probability, distributes the transmissions of users among a larger number of time-slots.

A more detailed comparison of the schemes is given in the following figures. 
Specifically, Fig.\ref{DMRFP} and Fig.\ref{DMRDP} illustrate the beneficial use of the MIAD power control for the detection miss ratio, which leads to a drastic reduction of the average number of miss-detected signals in the system for DPDB protocols. Obviously the miss-detection curves for FPDB are similar to the FPFB case, since no power control is applied. Furthermore, considering the contention ratio CR, both Fig.\ref{CRFP} and Fig.\ref{CRDP} show benefits compared to the fixed FPFR case. Interestingly, the DPDB cases are slightly worse than the FPDB. This is because a higher number $N_d\left(t\right)$ is detected for the same window size $W$, so that the CR calculated as in (\ref{DMRa}) appears higher.

\subsection{Protocol adaptation to channel fluctuations and deep fades}

In the current subsection, we further illustrate the performance of our protocol - which operates with parameters given in Table \ref{parameterB} - for a scenario with fluctuations and abrupt changes of the fading conditions. Such investigation shows how fast and with which cost in power expenditure can the protocol adapt to environmental changes. Specifically, we use a factor $\beta$ to multiply the long-term fading of each user. Initially the factor has an expectation $1$ and its value fluctuates uniformly within the interval $\left[0.7, 1.3\right]$. After a certain time-interval 
we initiate a sudden deterioration of the channel to an average of $0.8$, which returns to $1$ after some time. The realization of such fading scenario for a given user is presented in Fig. \ref{fig:Fluc}.

Very important here is to show how the protocol performs over time and adapts to the changes. Compared to the fixed power scenario, our suggested protocol can react very fast to the changes by an increase in power consumption during the period of the deep fade, which keeps the DMR always within the defined interval $\left[DMR^L,DMR^H\right]$. This can be observed in Fig.\ref{Pow1} and Fig.\ref{PowDMR1}.

           %


\section{Conclusions}
\label{Section5}

We have suggested a dynamically adaptive protocol which updates the user access probabilities and transmission powers 
in cellular random access communications for LTE systems. The protocol is based on measurements and user reports at the base station 
side, which allow for an estimation of the number of users present within the cell, as well as the quantities of detection-miss and contention probability. The protocol updates take place per time slot in a myopic fashion. By solving a drift minimization problem for the contention level 
and using closed loop updates for the transmission power level by a MIAD rule, the BS coordinates 
the actions chosen by the users, by broadcasting the pair $\left(L^*\left(t\right),p^*\left(t\right)\right)$.

The algorithmic steps, together with the methodology of the drift minimization for a certain measure of interest
referring to the steady state, provide a general suggestion to treat problems of self-organization in wireless
networks. Considering the specific scheme, a large variation of algorithms can be extracted, 
by choosing e.g. some different state function for the performance measure, 
or by introducing other kinds of user reports, which may provide more information to the central receiver, at the 
cost of increase in signaling. Furthermore, a larger action set can definitely provide 
a higher performance, compared to the proposed one - which introduces two possible 
values for the contention level (high/low) and two actions for the power level (increase/decrease). 
Even in this scheme however, which is characterized by an ``economy'' of signaling and information 
exchange, the results - as illustrated by numerical examples - are extremely beneficial, especially as the user number in the cell increases.

\newpage







\section*{Appendix - Relation between the Drift Minimization and a Markov Decision Problem Solution}

We begin this section by considering an ideal setting, meaning that all expressions are known and the system is fully controllable by the choice of actions. Let $V\left(\mathbf{S}\left(t\right)\right)$ be a non-negative function of the system state 
and let $\mathcal{M}\left(V,\tilde{\mathbf{A}}\right)$ be a performance metric related to the steady state reached when $t\rightarrow\infty$, if the initial state is $\mathbf{S}\left(0\right)$. The metric is a function of the entire set of actions $\tilde{\mathbf{A}}$
\begin{eqnarray}
 \label{PerfMeas}
\mathcal{M}\left(V,\tilde{\mathbf{A}}\right) & := & \lim_{t\rightarrow\infty}\mathbb{E}\left[V\left(\mathbf{S}\left(t\right)\right)|\mathbf{S}\left(0\right)\right].
\end{eqnarray}
If the actions are chosen per time-slot $t$ from the set $\mathbf{A}\left(t\right)$, the following general MDP can be posed:
\begin{eqnarray}
 \label{ProblemM}
\begin{tabular}{l l}
\textbf{min} 					& $\mathcal{M}\left(V,\tilde{\mathbf{A}}\right)$\\
\textbf{s.t.}					& $\mathbf{A}\left(t\right)\in\mathbb{A}$, $t=0,1,\ldots$
\end{tabular}
\end{eqnarray}

\begin{Pro}
 \label{Pro3}
The MDP in (\ref{ProblemM}) can be solved using the dynamic programming tools. The optimal solution satisfies Bellman's equation \cite{Puterman}
\begin{eqnarray}
 \label{Bellman}
J\left(\mathbf{S}\right) & = & \min_{\mathbf{A}\in\mathbb{A}}\left\{D\left(V\left(\mathbf{S}\right),\mathbf{A}\right) + \sum_{\mathbf{S}\acute{}\in\mathcal{S}}p_{s\rightarrow s\acute{}}J\left(\mathbf{S}\acute{}\right)\right\},\ \forall \mathbf{S}\in \mathcal{S}
\end{eqnarray}
 for the cost-to-go function $J\left(\mathbf{S}\right)$, where $\mathbf{S}\acute{}$ is the possible state at the next time slot, 
while the transition probabilities $p_{s\rightarrow s\acute{}}$ are functions of the actions chosen. The solution is state-dependent, meaning that the 
optimal actions depend on the system state and not on time.
\end{Pro}

\begin{Cor}
 \label{Pro4}
The solution of the drift minimization problem (\ref{ProblemG}) at each time slot $t$, is a suboptimal solution to the MDP in (\ref{ProblemM}). It is called one-stage look-ahead (myopic), in the sense that the actions are chosen per slot, considering only the transition to the next state and not the entire cost-to-go. 
\end{Cor}

\subsection*{Proof of \textbf{Proposition \ref{Pro3}}}
We first need the following lemma
\begin{Lem}
 \label{Pro2}
The performance measure can be written as an infinite sum of expected drifts over the discrete time axis, given the initial state $\mathbf{S}\left(0\right)$
\begin{eqnarray}
 \label{SumDrift}
\mathcal{M}\left(V,\tilde{\mathbf{A}}\right) = V\left(\mathbf{S}\left(0\right)\right) + \sum_{t=0}^{\infty}\mathbb{E}\left[D\left(V\left(\mathbf{S}\left(t\right)\right),\mathbf{A}\left(t\right)\right)|\mathbf{S}\left(0\right)\right].
\end{eqnarray}
\end{Lem}
\begin{proof}: Let $\mathcal{F}^{(t)}:=\left\{\mathbf{S}\left(0\right),\ldots,\mathbf{S}\left(t\right)\right\}$ be the information over the 
system realizations up to slot $t$. Obviously $\mathcal{F}^{(0)}\subseteq \mathcal{F}^{(t)}$ (formally we call $\left\{\mathcal{F}^{(t)},\ t\geq 0\right\}$ a filtration and $\mathcal{F}^{(0)}$ is a sub-$\sigma$-algebra of $\mathcal{F}^{(t)}$) 
and the tower property for expectations \cite[p.88]{Williams} holds. Hence,
 \begin{eqnarray}
  \mathbb{E}\left[V\left(\mathbf{S}\left(t+1\right)\right)|\mathbf{S}\left(0\right)\right] 	& \stackrel{Tower}{=} 				& \mathbb{E}\left[\mathbb{E}\left[V\left(\mathbf{S}\left(t+1\right)\right)|\mathcal{F}^{(t)}\right]|\mathcal{F}^{(0)}\right]\nonumber\\
												& \stackrel{Markov}{=} 		& \mathbb{E}\left[\mathbb{E}\left[V\left(\mathbf{S}\left(t+1\right)\right)|\mathbf{S}{(t)}\right]|\mathbf{S}{(0)}\right]\nonumber\\
												& \stackrel{(\ref{Drift})}{=}				& \mathbb{E}\left[D\left(V\left(\mathbf{S}\left(t\right)\right),\mathbf{A}\left(t\right)\right)|\mathbf{S}\left(0\right)\right] + \mathbb{E}\left[V\left(\mathbf{S}\left(t\right)\right)|\mathbf{S}{(0)}\right]\nonumber
 \end{eqnarray}
and by repeating the process for $t,\ldots,0$ and taking the limits for $t\rightarrow \infty$ we reach the result.
\end{proof}

Now we can continue with the proof of the the Proposition. Consider the series in (\ref{SumDrift}) up to a finite horizon $T+1$ and denote the related sum by $\mathcal{M}_T\left(V,\tilde{\mathbf{A}}\right)$. Then the expected drift term for some $\tau\leq T$ equals 

\begin{eqnarray}
 \label{DriftT}
\mathbb{E}\left[D\left(V\left(\mathbf{S}\left(\tau\right)\right),\mathbf{A}\left(\tau\right)\right)|\mathbf{S}\left(0\right)\right] & = & \nonumber\\
\sum_{\mathbf{S}\left(1\right)}\ldots\sum_{\mathbf{S}\left(\tau\right)}p_{s_o\rightarrow s_1}\ldots p_{s_{\tau-1}\rightarrow s_{\tau}} D\left(V\left(\mathbf{S}\left(\tau\right)\right),\mathbf{A}\left(\tau\right)\right) & & \nonumber
\end{eqnarray}

It can be observed that $p_{s_{\tau-1}\rightarrow s_{\tau}}$, which can be controlled by the actions $\mathbf{A}\left(\tau-1\right)$ appear in all summands of $\mathcal{M}_T\left(V,\tilde{\mathbf{A}}\right)$, for $\tau \leq \hat{t}\leq T$ and not for $0\leq t\leq \tau-1$. 
Following this observation, the optimal choice of actions $p_{s_{T}\rightarrow s_{T+1}}^*$ are found by solving $\min_{\mathbf{A}\left(T\right)\in\mathbb{A}}$ $\mathcal{M}_T\left(V,\tilde{\mathbf{A}}\right)$, the cost-to-go at $T$. 

The cost-to-go can be verified to satisfy the recursion, $\forall \mathbf{S}\left(\tau-1\right)\in\mathcal{S}$:
\begin{eqnarray}
 \label{CostToGoT}
J\left(\mathbf{S}\left(\tau-1\right)\right) & = & \min_{\mathbf{A}\left(\tau-1\right)\in\mathbb{A}}\sum_{\mathbf{S}\left(\tau\right)}p_{s_{\tau-1}\rightarrow s_{\tau}} \left(V\left(\mathbf{S}\left(\tau\right)\right)-V\left(\mathbf{S}\left(\tau-1\right)\right)+J\left(\mathbf{S}\left(\tau\right)\right)\right).\nonumber
\end{eqnarray}
The expression holds as well, when we let the horizon $T\rightarrow\infty$. Thus taking $\tau\rightarrow \infty$ results in (\ref{Bellman}).

\newpage

\section*{TABLES}

\begin{table}[h]
 \label{Algo1}
\centering
\caption{GENERAL SELF-OPTIMIZATION ALGORITHM}
\begin{tabular}[c]{l l l}
\hline
STEP 1 & \vline & Gather empirical information $\mathcal{I}$ at the BS.\\
STEP 2 & \vline & Estimate unknown factors (see 1. - 3. above).\\
STEP 3 & \vline & Solve the resulting optimization problem in (\ref{ProblemG}).\\
STEP 4 & \vline & Broadcast action-related information $\mathcal{J}$.\\
STEP 5 & \vline & Calculate at the user side the required actions, based on $\mathcal{J}$.\\
\hline
\end{tabular}
\label{RACHalgo}
\end{table}

\begin{table}[th]
      \centering
      \caption{PARAMETER TABLE}
      \begin{tabular}[c]{l l c}
      \hline
      Parameters & \vline & Value\\
      \hline
      Wireless Network & \vline & Single cell\\
      User distribution & \vline & Uniform within cell \\
      Number of users in cell & \vline & $\left\{1,2,\ldots,14\right\}$ \\
      Sequence pool size & \vline & $10$\\
      Fixed Tx Power & \vline & $250$ mW\\
      Power ramping step $\Delta p$ & \vline & $20$ mW\\
      Maximum Tx Power & \vline & $500$ mW\\
      Path loss $PL$ & \vline & $128.1+37.6\log(D \ km)$ dB\\
      Noise & \vline & $-133.2$ dBm\\
      SNR threshold & \vline & $8$ dB \\
      Maximum effort $M$ & \vline & $5$ \\
      Fixed backoff probability & \vline & $[0.5, 0.4, 0.3, 0.2, 0.1]$ \\     
      Number of slots & \vline & $15000$ slots\\
       \hline
       \end{tabular}
       \label{parameterA}
\end{table}

\begin{table}[th]
      \centering
      \caption{TUNABLE FACTORS TABLE}
      \begin{tabular}[c]{llc}
      \hline
      Tunable Factors & \vline & Value\\
      \hline
      Window length $W$ & \vline & $200$ slots\\ 
      Backoff factor $A$ & \vline & $\left\{0.05, 0.25,0.5\right\}$\\
      Access Function $f\left(m\right)$ & \vline & $m^{-1}$\\
      Power control factor $\delta_1$ & \vline & $2\times 10^{-4}$\\
      Power control factor $\delta_2$ & \vline & $8$ mW\\
       $DMR^H$ & \vline & $3.5\%$ \\
       $DMR^L$ & \vline & $2.5\%$ \\
       \hline
       \end{tabular}
       \label{parameterB}
\end{table}

\newpage

\section*{FIGURES}

\begin{figure}[ht]
\centering
           \includegraphics[width=0.75\textwidth]{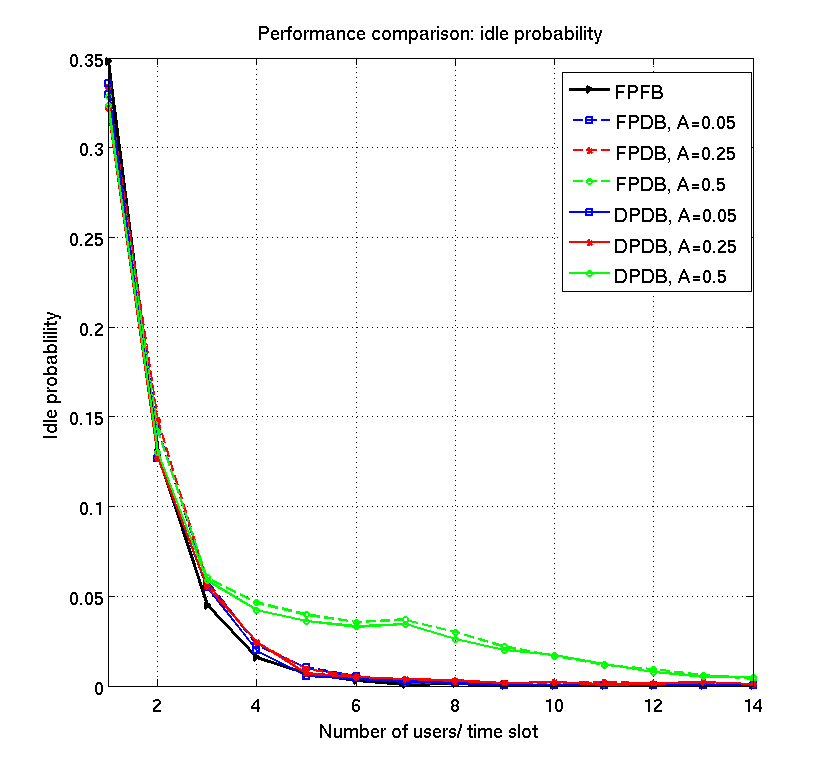}
           \caption{Comparison of the average occurence of idle slot per scheme. The dynamic scenario with $\mathcal{A}=0.05$ is the closest to follow the chosen fixed one.}
           \label{idle}
\end{figure}

\begin{figure}[ht]
\centering
           \includegraphics[width=0.75\textwidth]{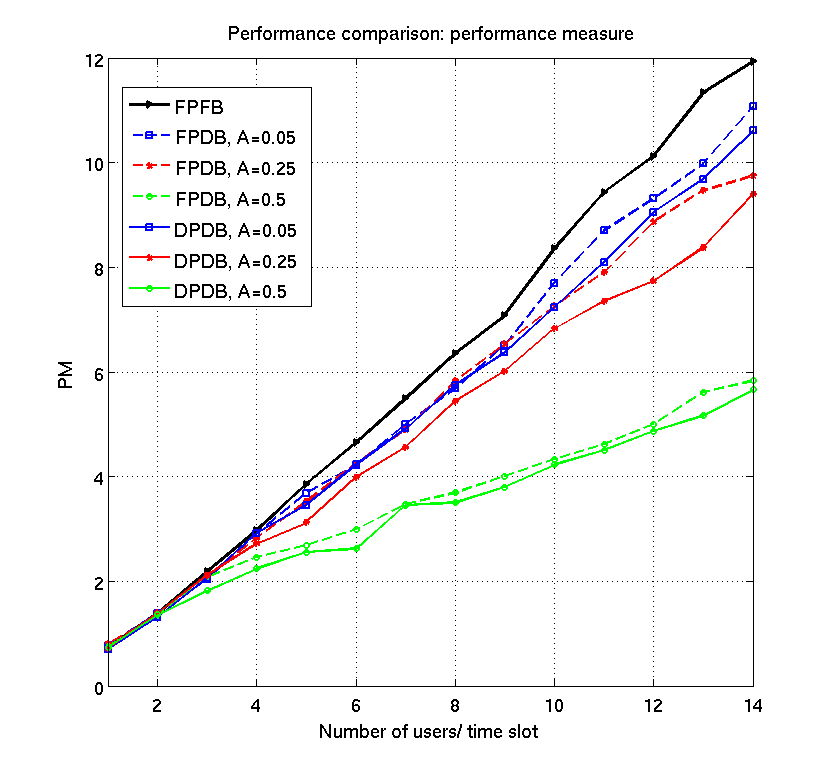}
           \caption{Comparison of performance measure, equal to the chosen function $V$ as $t\rightarrow \infty$. The measure improves with increasing 
idle probability bound $\mathcal{A}$. Furthermore, all DPDB schemes outperform the FPDB ones.}
           \label{PM}
\end{figure}

\begin{figure}[ht]
\centering
           \includegraphics[width=0.75\textwidth]{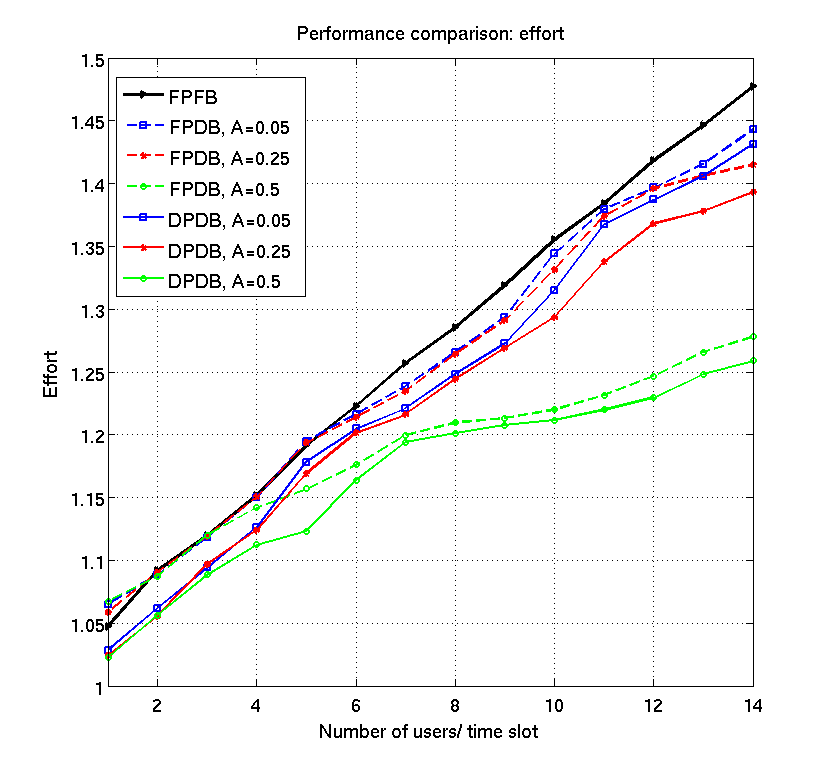}
           \caption{Comparison of the average number of efforts until success. The behaviour of these curves follows closely the performance metric curves, due to the specific choice of 
           the Lyapunov function $V$ as sum of user states.}
           \label{effort}
\end{figure}

\newpage

\begin{figure}[ht]
\centering
\subfigure[Total delay in FPDB protocols.]{
\includegraphics[width=0.45\textwidth]{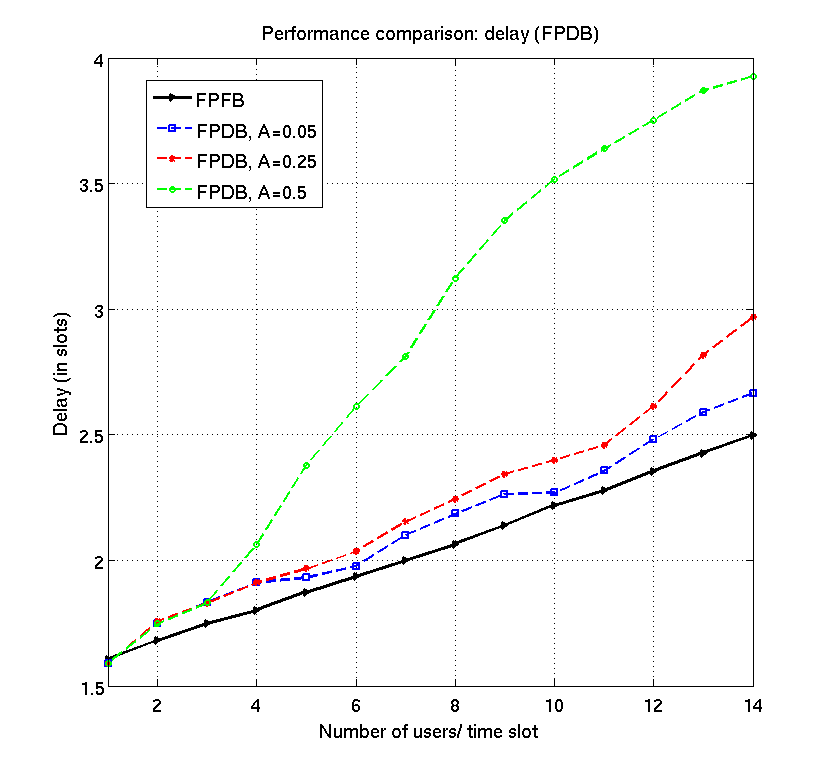}
\label{delayFP}
} 
\subfigure[Total delay in DPDB protocols.]{
\includegraphics[width=0.45\textwidth]{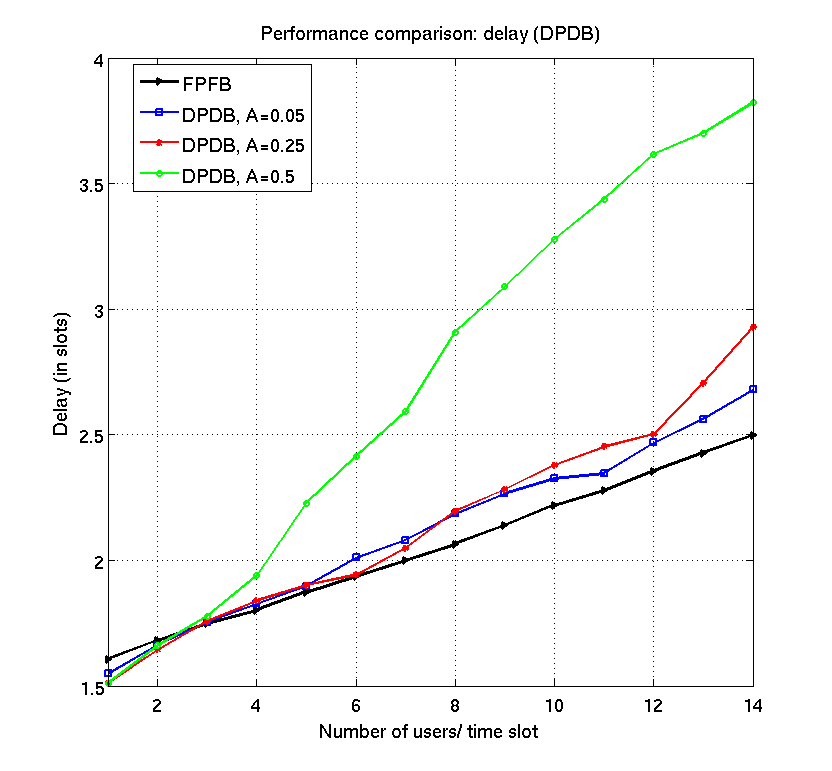}
\label{delayDP}
} 
\label{Two}
\caption{Evaluation of total average delay up to success (including backoff slots) in the case of (a) FPDB protocols and (b) DPDB protocols. The higher the 
parameter $\mathcal{A}$, the higher the allowed delay. For $\mathcal{A}=0.05$, the protocol delay approaches the one of the FPFB protocol. In general power control improves the delay.}
\end{figure}

\begin{figure}[ht]
\centering
\subfigure[Tx power in FPDB protocols.]{
\includegraphics[width=0.45\textwidth]{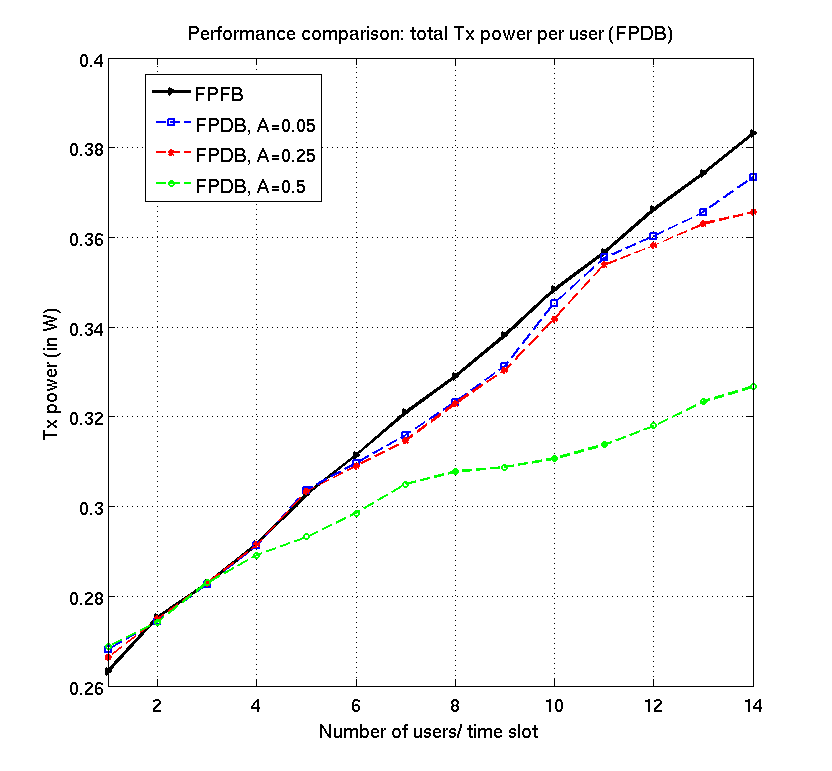}
\label{PFP}
} 
\subfigure[Tx power in DPDB protocols.]{
\includegraphics[width=0.45\textwidth]{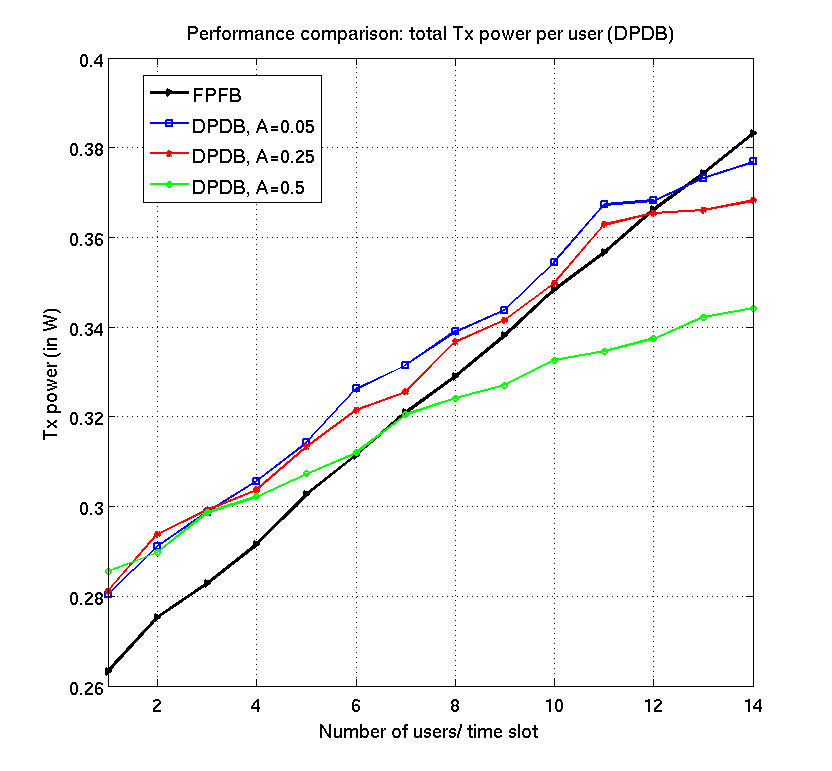}
\label{PDP}
} 
\label{Three}
\caption{Evaluation of average Tx Power consumption up to success in the case of (a) FPDB protocols and (b) DPDB protocols. In the 
case of FPDB, the consumed power is always lower than the FPFB case. Both cases exhibit benefits in Tx power.}
\end{figure}

\newpage

\begin{figure}[ht]
\centering
\subfigure[Dropping rate in FPDB protocols.]{
\includegraphics[width=0.45\textwidth]{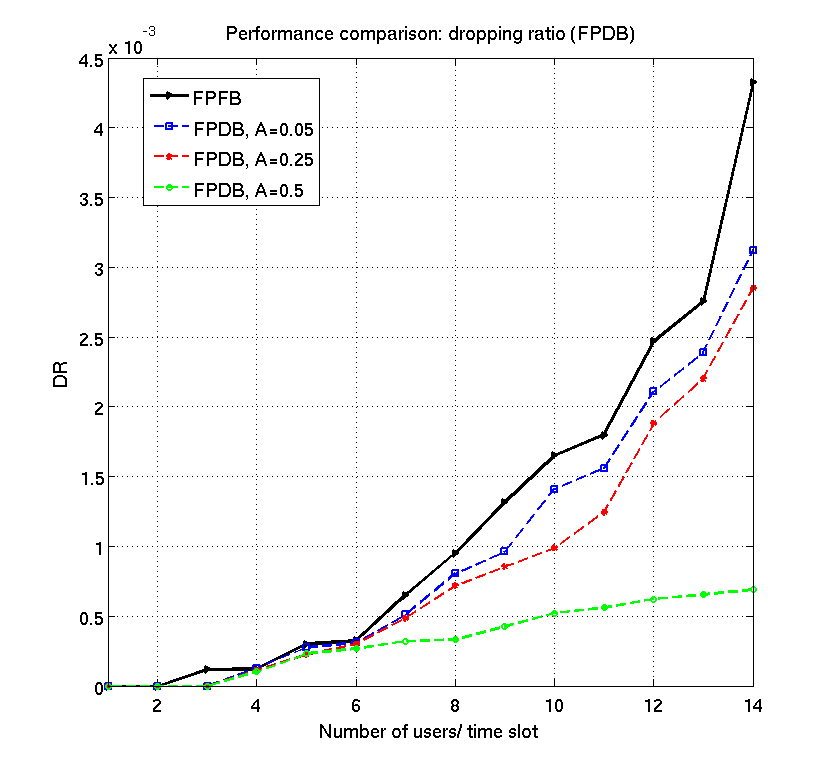}
\label{DRFP}
} 
\subfigure[Dropping Rate in DPDB protocols.]{
\includegraphics[width=0.45\textwidth]{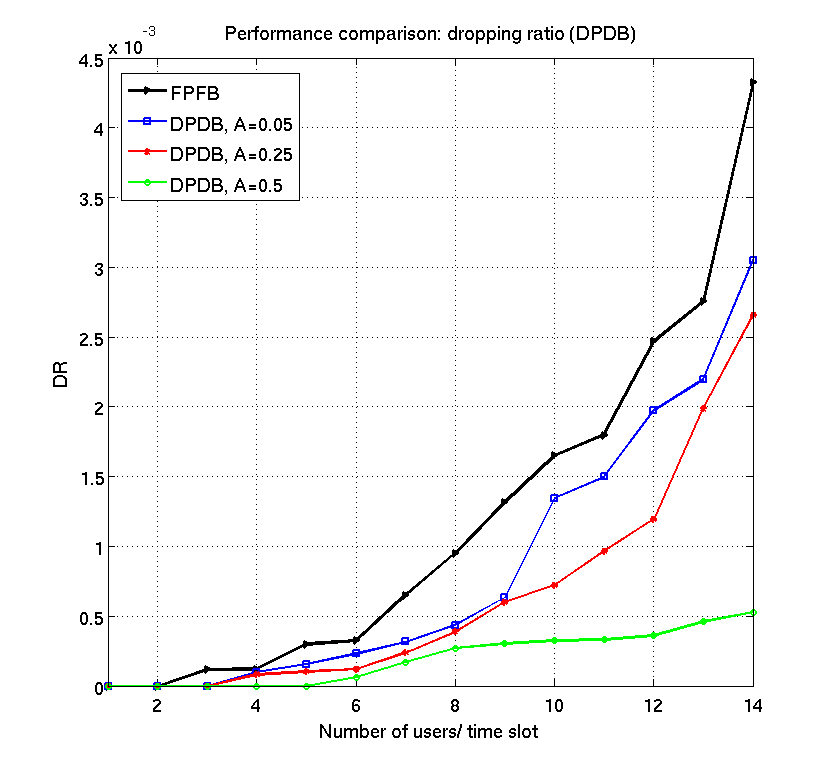}
\label{DRDP}
} 
\label{Four}
\caption{Comparison of the average dropping rate (DR) in the case of (a) FPDB protocols and (b) DPDB protocols.. The abrupt increase of the rate after a certain user number is an indicator
that the system is not anymore stable for a further increase in the cell user number. Higher values of $\mathcal{A}$ can increase the 
point when the instability appears, at the cost of delay. (For a single user, the dropping rate may be non-zero if the event of miss-detection occurs $M$ consecutive 
times due to bad channel conditions and poor transmission power.)}
\end{figure}

\newpage

\begin{figure}[ht]
\centering
\subfigure[Miss-detection rate in FPDB.]{
\includegraphics[width=0.45\textwidth]{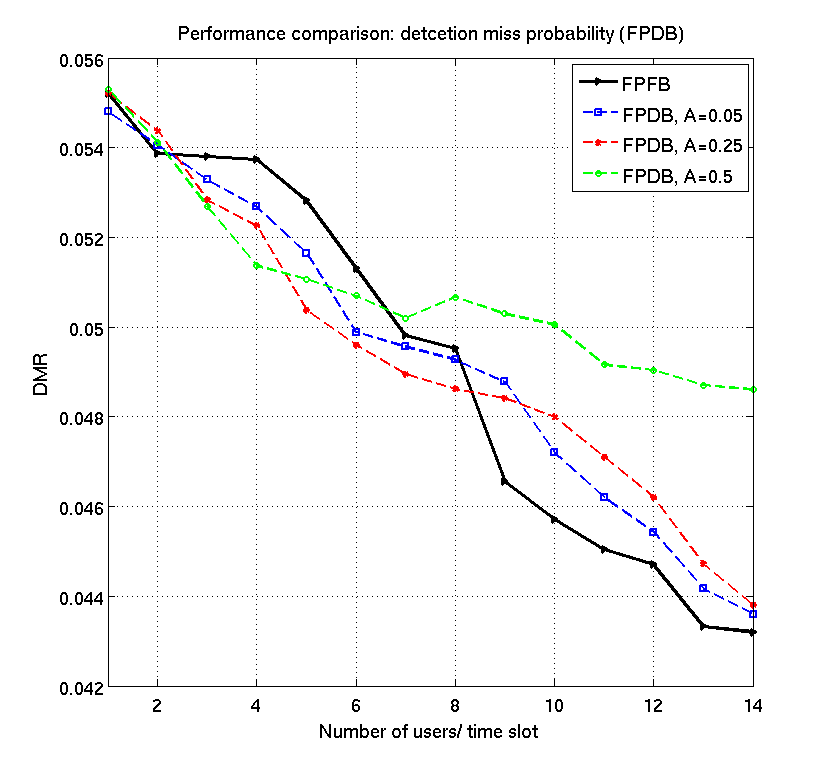}
\label{DMRFP}
} 
\subfigure[Miss-detection rate in DPDB]{
\includegraphics[width=0.45\textwidth]{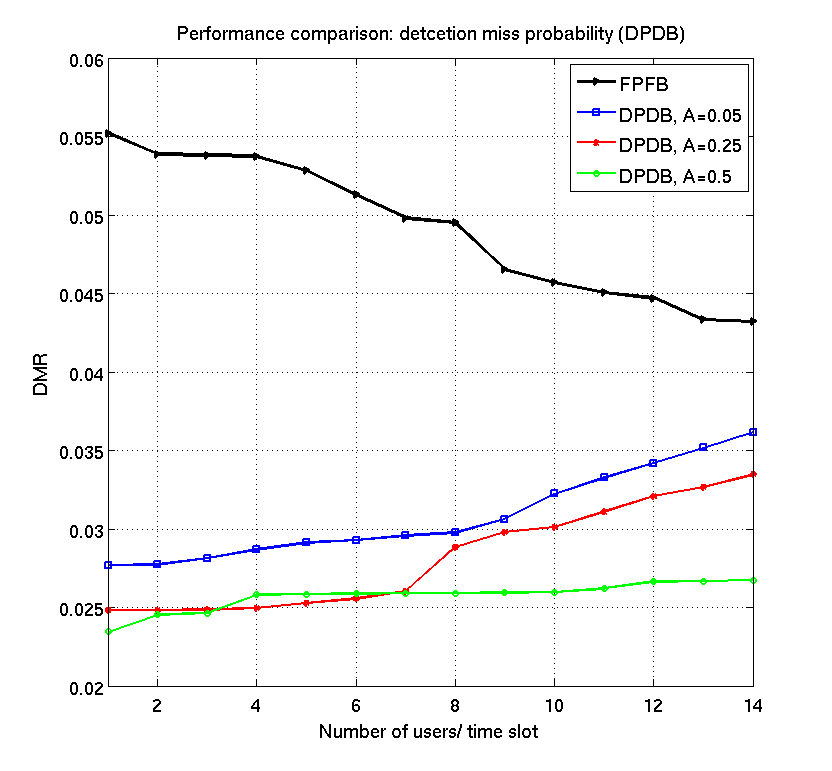}
\label{DMRDP}
} 
\label{Five}
\caption{Comparison of miss-detection rate DMR for the two protocols (a) FPDB and (b) DPDB. Benefits are evident only in the case (b) where the MIAD rule is applied.}
\end{figure}

\begin{figure}[ht]
\centering
\subfigure[Contention rate rate in FPDB.]{
\includegraphics[width=0.45\textwidth]{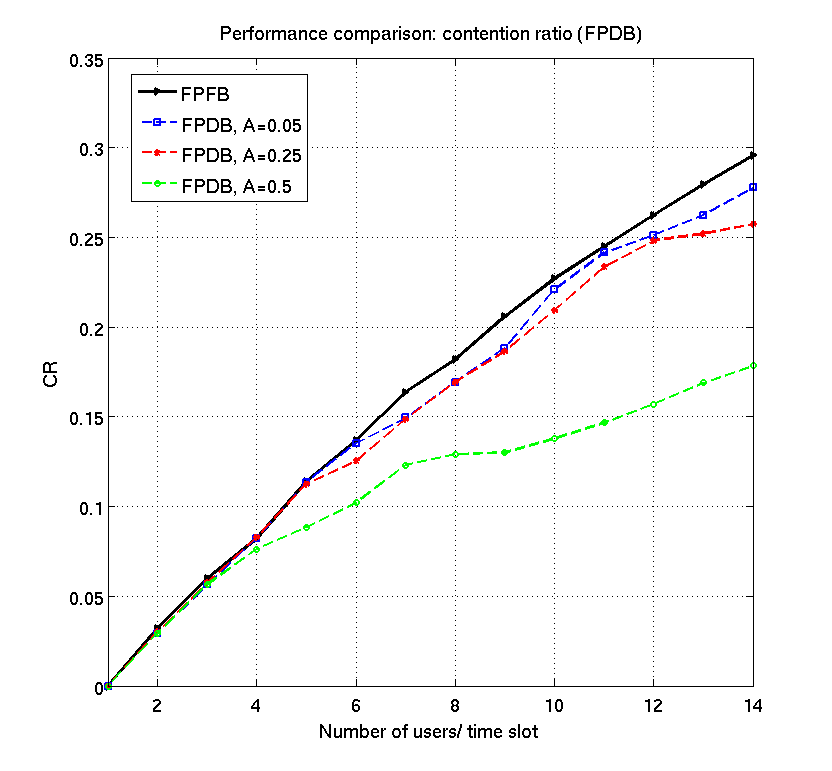}
\label{CRFP}
} 
\subfigure[Contention rate in DPDB]{
\includegraphics[width=0.45\textwidth]{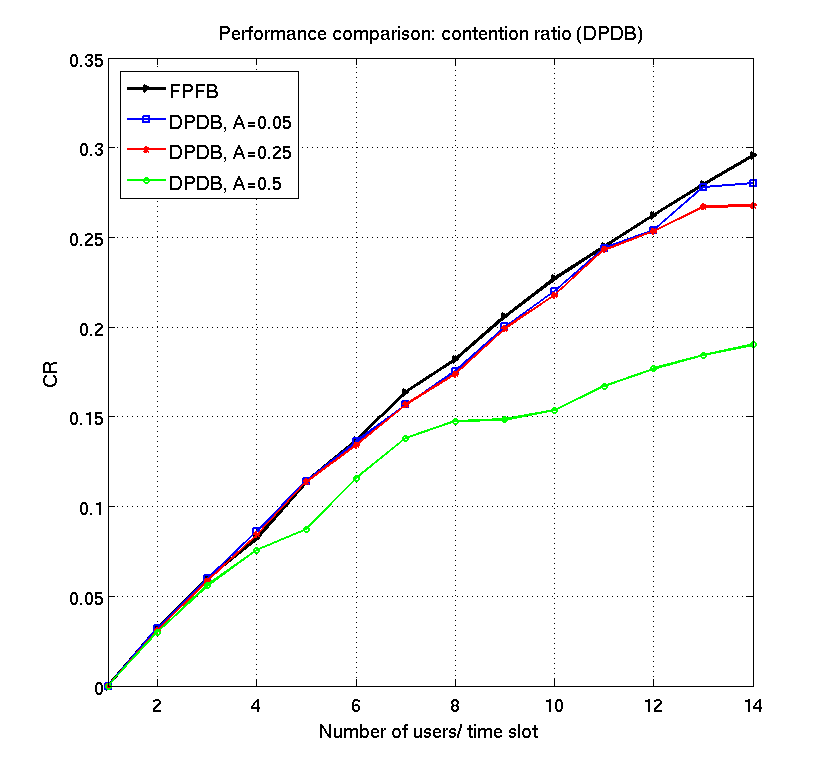}
\label{CRDP}
} 
\label{Six}
\caption{Comparison of contention rate CR for the two protocols (a) FPDB and (b) DPDB. Both schemes exhibit improvements compared to the FPFB case, due to 
the backoff optimal choices. The case DPDB is slightly worse than the FPDB due to the fact that a larger number of packets are detected, so that the CR appears lower.}
\end{figure}

\newpage

\begin{figure}[ht]    
\centering  
	 \subfigure[Scenario with channel fluctuations and deep fades.]{          
           \includegraphics[width=0.45\textwidth]{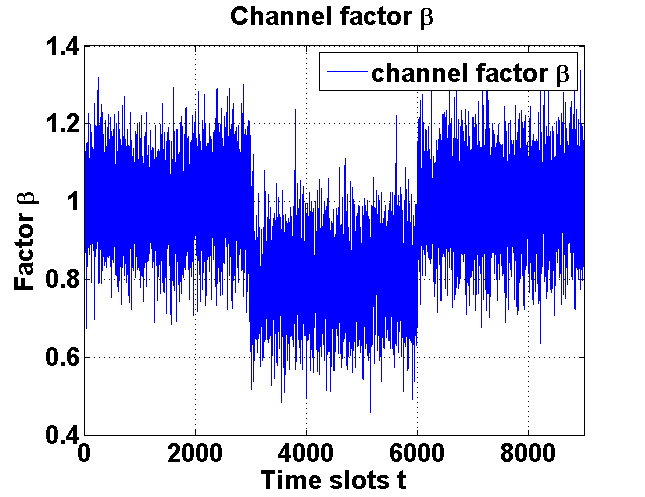}
           \label{fig:Fluc}
           }
            \subfigure[Temporal adaptation of transmission power to a deep fade.]{          
           \includegraphics[width=0.45\textwidth]{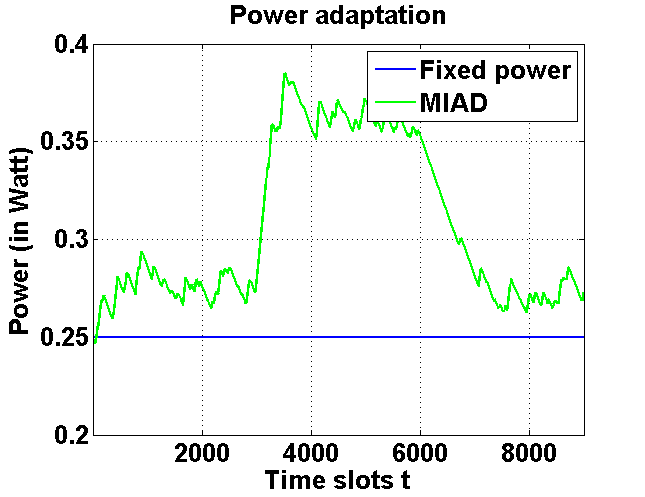}
           \label{Pow1}
           }
           \subfigure[Temporal variation of the DMR.]{          
           \includegraphics[width=0.45\textwidth]{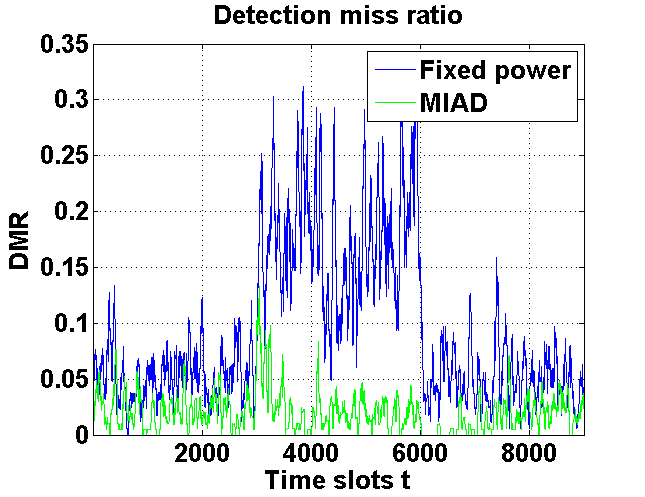}
           \label{PowDMR1}
           }
           \label{Temp}
           \caption{Protocol adaptation with respect to power and DMR}
\end{figure}

\end{document}